\documentclass[12pt]{article}
\usepackage[inner=2.5 cm,outer=2.0 cm,top=2.0 cm,bottom=2.0 cm]{geometry}

\usepackage{setspace}
\usepackage{geometry}
\usepackage{amsmath,amssymb,amsfonts,amsthm}
\usepackage[english]{babel}
\usepackage[latin1]{inputenc}
\usepackage{times}
\usepackage[T1]{fontenc}
\usepackage{graphicx}
\usepackage{amsfonts}
\usepackage[all]{xy}

\numberwithin{equation}{section}

\newcounter{item}[section]
\newcounter{kirshr}
\newcounter{kirsha}
\newcounter{kirshb}

\newtheorem{theorem}{Theorem}[section]

\newtheorem{lemma}[theorem]{Lemma}
\newtheorem{corollary}[theorem]{Corollary}

\theoremstyle{definition}

\newtheorem{definition}[theorem]{Definition}

\def\(R)RA{{\bf (R)RA}}

\title{On Symmetrized Weight Compositions}
\author{Ali Assem\\\small{Department of Mathematics}\\\small{Faculty of Science}\\\small{Cairo University}\\\small{aassem@sci.cu.edu.eg}}
\begin{document}
\maketitle
\begin{abstract}A characterization of module alphabets with the Hamming weight EP (abbreviation  for Extension Property) had been settled. A thoughtfully constructed piece-of-art example by J.A.Wood (\cite{r6}) finished the tour. In 2009, in \cite{r7},  Frobenius bimodules were proved to satisfy the EP with respect to \emph{symmetrized weight compositions}. In \cite{r4}, the embeddability in the character group of the ambient ring $R$ was found sufficient for a module $_RA$ to satisfy the EP with respect to swc  built on any subgroup of $\mathrm{Aut}_R(A)$, while the necessity remained a question.

\setlength{\parindent}{0cm}Here, landing in a ``Midway'',  the necessity is proved by jumping to Hamming weight. Corollary \ref{cor2} declares a characterization of module alphabets satisfying the EP with respect to swc. 
 \end{abstract}
\textbf{Note:} All rings are finite with unity, and all modules are finite too. This may be re-emphasized in some statements. The convention for functions is that inputs are to the left.

\setcounter{section}{1}

\section*{Symmetrized Weight Compositions}

\setlength{\parindent}{0.4cm}

\begin{definition}(Symmetrized Weight Compositions) Let $G$ be a subgroup of the automorphism group of a finite $R$-module $A$. Define $\sim$ on $A$ by $a\sim b$ if $a=b\tau$ for some $\tau \in G$. Let $A/G$ denote the orbit space of this action. The symmetrized weight composition is a function

swc : $A^n \times A/G\rightarrow \mathbb{Q}$ defined by, $$\mathrm{swc}(x, a) = \mathrm{swc}_a(x) = |\{i:x_i\sim a\}|,$$where $x=(x_1,\ldots,x_n)\in A^n$ and $a\in A/G$.\end{definition}

\begin{definition} Let $G$ be a subgroup of $\mathrm{Aut}_R(A)$, a map $T$ ia called a $G$-monomial transformation of $ A^n$ if for any $(x_1,\ldots,x_n)\in A^n$ $$(x_1,\ldots,x_n)T=(x_{\sigma(1)}\tau_1,\ldots,x_{\sigma(n)}\tau_n),$$where $\sigma\in S_n$ and $\tau_i\in G$ for $i=1,\ldots,n$.\end{definition}

\begin{definition}(Extension Property)  The alphabet $A$ has the extension
property with respect to swc if for every $n$, and any two linear codes $ C_1, C_2\subset A^n $,  any
$R$-linear isomorphism $ f:C_1 \rightarrow C_2$  that preserves swc is extendable to a
$G$-monomial transformation of $ A^n$.\end{definition}

In \cite{r5}, J.A.Wood proved that Frobenius rings do have the extension property with respect to swc. Later, in \cite{r4}, it was shown that, more generally, a left $R$-module $A$ has the extension property with respect to swc if it can be embedded in the character group $\widehat{R}$ (given the standard module structure).

\begin{theorem}(Th.4.1.3, \cite{r3}) Let $A$ be a finite left $R$-module. If A can be embedded into  $\widehat{R}$, then $A$ has the extension
property with respect to the swc built on any subgroup $G$ of $\mathrm{Aut}_R(A)$. In particular, this theorem applies to Frobenius bimodules.\end{theorem}

We now define a new notion (the Midway!) on which we'll depend in the rest of this paper.
\begin{definition}(Annihilator Weight) On $_RA$, define an equivalence relation $\rho$ by $a\rho b$ if $\mathrm{Ann}_a=\mathrm{Ann}_b$, where $a$ and $b$ are any two elements in $A$ and $\mathrm{Ann}_a=\{r\in R{}|{} ra=0\}$ is the annihilator of $a$. Clearly, $\mathrm{Ann}_a$ is a left ideal.

Now, on $A^n$ we can define the annihilator weight $aw$ that counts the number of components in each orbit. \end{definition}

\textbf{Remark:} It is easily seen that  the EP  with respect to Hamming weight implies  the EP with respect to swc built on $\mathrm{Aut}_R(A)$, and the EP with respect to $aw$ as well.

\begin{lemma}\label{lemma1} Let $_RA$ be a pseudo-injective module. Then for any two elements $a$ and $b$ in $A$, $a\rho b$ if and only if $a\sim b$.\end{lemma}

\begin{proof}If $a\sim b$, this means $a=b\tau$ for some $\tau \in \mathrm{Aut}_R(A)$, and consequently $\mathrm{Ann}_a=\mathrm{Ann}_b$.

Conversely, if $a\rho b$, then we have (as left $R$-modules) $$Ra\cong{} _RR/\mathrm{Ann}_a={} _RR/\mathrm{Ann}_b\cong Rb,$$with $ra\mapsto r+\mathrm{Ann}_a\mapsto rb$. By Proposition 5.1 in \cite{r7}, since $A$ is pseudo-injective, the isomorphism $Ra\rightarrow Rb\subseteq A$ extends to an automorphism of $A$ taking $a$ to $b$.\end{proof}

\begin{corollary}\label{cor1}If $A$ is a pseudo-injective module, then the EP with respect to $swc$ built on $\mathrm{Aut}_R(A)$ is equivalent to the EP with respect to $aw$.\end{corollary}

\begin{theorem}If $R$ is a chain ring, then for any pseudo-injective module $A$, the following are equivalent:
 \begin{enumerate}\item $A$ has the EP with respect to swc built on $\mathrm{Aut}_R(A).$\item $A$ has the EP with respect to Hamming weight.
  \end{enumerate}\end{theorem}

\begin{proof} First recall that the left ideals of a chain ring form a chain with inclusion, so we may assume that our chain is $I_0=\mathrm{rad}R\supset I_1\supset I_2\supset\ldots\supset I_m\supset(0)$.

By the remark above we know that (2) implies (1). For the converse, by Corollary \ref{cor1}, it's enough to show that if $A$ has the EP with respect to $aw$, then it has the EP with respect to Hamming weight.\\Suppose the EP with respect to $aw$ holds, and that $f:C\rightarrow A^n$ is a monomorphism preserving Hamming weight, let $(c_1,c_2,\ldots,c_n) \in C$ and $(c_1,c_2,\ldots,c_n)f=(b_1,b_2,\ldots,b_n)$.

Choose an element $x_0 \in I_0$ that doesn't belong to any of the smaller left ideals. Then, in the equality$$(x_0c_1,x_0c_2,\ldots,x_0c_n)f=(x_0b_1,x_0b_2,\ldots,x_0b_n),$$ the number of annihilated components is exactly the number of those $c_i$'s (and $b_i$'s)  with annihilator $I_0$. The preservation of Hamming weight then gives that this number is the same on each side. Repeating this process we get that $f$ preserves $aw$ and hence extends to a monomial transformation.
\end{proof}

\begin{theorem}Let $R$ be the ring $\mathbb{M}_m(\mathbb{F}_q)$ of square matrices of size $m\times m$ over a the finite field $\mathbb{F}_q$ , and let $A$ be a finite $R$-module. Then the EP with respect to swc built on $\mathrm{Aut}_R(A)$ holds if and only if the EP with respect to Hamming weight holds.\end{theorem}

\begin{proof} The ``if'' part is direct. For the converse, notice that $A$ is now injective (being a matrix module), so, again, it's enough to prove that if $A$ has the EP with respect to $aw$, then it has the EP with respect to Hamming weight.\\Suppose the EP with respect to $aw$ holds, and that $f:C\rightarrow A^n$ is a monomorphism preserving Hamming weight.

Let's first focus on $R=\mathbb{M}_m(\mathbb{F}_q)$. $R$ is a finite simple ring (hence semisimple and left artinian), thereby, any left ideal $I$ has the form $Re_I$, where $e_I$ is an idempotent (Theorem ix.3.7, \cite{rw}).

It is clear, then, that any left ideal $I$ contains an element $e_I$ that does not belong to any other left ideal not containing $I$. Now, if \begin{equation}\label{eq1}(c_1,c_2,\ldots,c_n)f=(b_1,b_2,\ldots,b_n),\end{equation} choose, from $c_1,c_2,\ldots,c_n;b_1,b_2,\ldots,b_n$, a component with a maximal annihilator  $I$. Act on equation (\ref{eq1}) by $e_I$, then the only zero places  are those of the components in equation (\ref{eq1}) with annihilator $I$, the preservation of Hamming weight gives the preservation of $I$-annihilated components. Omit these components from the list $c_1,c_2,\ldots,c_n;b_1,b_2,\ldots,b_n$ and choose one with the new maximal, and repeat. This gives that $f$ preserves $aw$ and hence extends to a monomial transformation.
\end{proof}

\begin{corollary}Let $R$ be the ring $\mathbb{M}_m(\mathbb{F_q})$, and $A=\mathbb{M}_{m,k}(\mathbb{F_q})$ where $k>m\;$, then $A$ doesn't have the EP  with respect to swc.\end{corollary}

\begin{proof} In \cite{r6}, J.A.Wood proved that $A$ doesn't have the Hamming weight EP, hence, by the above $A$  cannot have the EP with respect to swc. \end{proof}

\begin{corollary}\label{cor2}If $_RA$ be a finite module over a finite ring $R$ with unity, then $A$ has the extension property with respect to swc if and only if can be embedded in $\widehat{R}$ (or equivalently $A$ has a cyclic socle).\end{corollary}

 \end{document}